%% file: main.tex
\documentclass[letterpaper,11pt]{article}

\usepackage{setspace}
\usepackage{subcaption}
\usepackage{enumerate}
\usepackage{amsmath}
\usepackage{amsfonts}
\usepackage{graphicx}
\usepackage{tcolorbox}
\usepackage{amssymb}
\usepackage{multirow}
\usepackage{geometry}
\usepackage{soul}
\usepackage{colortbl}
\usepackage{wrapfig}
\usepackage{algorithm}
\usepackage{algorithmic}
\usepackage{amsthm}
\usepackage{todonotes}
\usepackage{mathtools}
\usepackage{mathrsfs}

\usepackage[pdftex, plainpages = false, pdfpagelabels, 
                 bookmarks=false,
                 bookmarksopen = true,
                 bookmarksnumbered = true,
                 breaklinks = true,
                 linktocpage,
                 pagebackref,
                 colorlinks = true,  
                 linkcolor = blue,
                 urlcolor  = blue,
                 citecolor = red,
                 anchorcolor = green,
                 hyperindex = true,
                 hyperfigures
                 ]{hyperref}

\usepackage{thmtools} 
\usepackage{thm-restate}

\usepackage{tcolorbox}
\usepackage{xspace}

\newcommand{\colorof}{{\mathfrak{C}}}
\newcommand{\NN}{\text{NN}}
\newcommand{\mcs}{\text{MCS}\xspace}
\newcommand{\cspg}{{CSPG}\xspace} 
\newcommand{\fpt}{\text{FPT}\xspace}
\newcommand{\np}{\textup{\textsf{NP}}\xspace}
\newcommand{\npc}{\textsf{NP}-complete\xspace}
\newcommand{\bigoh}{\mathcal{O}}
\newcommand{\bigohstar}{\mathcal{O}^{*}}
\newcommand{\nd}{\text{ND}}
\newcommand{\DistSet}[2]{N^{#2}(#1)}
\newcommand{\DistSeta}[3]{N^{#2}_{#3}(#1)}
\newcommand{\U}{\mathcal{U}}
\newcommand{\Hset}{\mathcal{H}}
\newcommand{\HS}{\mathcal{HS}}
\newcommand{\D}{\mathcal{D}}
\newcommand{\F}{\mathcal{F}}
\newcommand{\R}{\mathcal{R}}
\newcommand{\T}{\mathcal{T}}
\newcommand{\rc}{\mathsf{c}}

\definecolor{corn}{HTML}{555555}
\tcbset{
  mydefstyle/.style={
    colback=corn!5,      
    colframe=corn!20,    
    coltitle=black,
    fonttitle=\bfseries,
    boxrule=0.5pt,
    arc=2mm,
    left=2mm,
    right=2mm,
    top=1mm,
    bottom=1mm,
    title={#1}
  }
}

\newtheorem{theorem}{Theorem}[section]
\newtheorem{observation}[theorem]{Observation} 
\newtheorem{lemma}[theorem]{Lemma} 
\newtheorem{proposition}[theorem]{Proposition}   
\newtheorem{claim}[theorem]{Claim}

\usepackage{wrapfig}
\usepackage{bm}

\title{Learning with Structure: Computing Consistent Subsets \\on Structurally-Regular Graphs}
\author{
	Aritra Banik\thanks{National Institute of Science, Education and Research, An OCC of Homi Bhabha National Institute, Bhubaneswar, India.}
	\and
  Mano Prakash Parthasarathi\thanks{North Carolina State University, Raleigh, NC, USA.}
  \and
  Venkatesh Raman\thanks{(Retd) The Institute of Mathematical Sciences, HBNI, Chennai, India.}
  \and
  Diya Roy\footnotemark[1]
  \and
  Abhishek Sahu\footnotemark[1]
}

\begin{document}

  \pagenumbering{gobble}
  \maketitle

  \begin{abstract}
    The Minimum Consistent Subset (\mcs) problem arises naturally in the context of supervised clustering and instance selection. In supervised clustering, one aims to infer a meaningful partitioning of data using a small labeled subset. However, the sheer volume of training data in modern applications poses a significant computational challenge. The \mcs problem formalizes this goal: given a labeled dataset $\mathcal{X}$ in a metric space, the task is to compute a smallest subset $S \subseteq \mathcal{X}$ such that every point in $\mathcal{X}$ shares its label with at least one of its nearest neighbors in $S$.

    Recently, the \mcs problem has been extended to \emph{graph metrics}, where distances are defined by shortest paths. Prior work has shown that \mcs remains NP-hard even on simple graph classes like trees, though an algorithm with runtime $\bigoh(2^{6c} \cdot n^6)$ is known for trees, where $c$ is the number of colors and $n$ the number of vertices. This raises the challenge of identifying graph classes that admit algorithms efficient in both $n$ and $c$.

    In this work, we study the Minimum Consistent Subset problem on graphs, focusing on two well-established measures: the vertex cover number ($vc$) and the neighborhood diversity ($nd$). Specifically, we design efficient algorithms for graphs exhibiting small $vc$ or small $nd$, which frequently arise in real-world domains characterized by local sparsity or repetitive structure. These parameters are particularly relevant because they capture structural properties that often correlate with the tractability of otherwise hard problems. Graphs with small vertex cover sizes are "almost independent sets", representing sparse interactions, while graphs with small neighborhood diversity exhibit a high degree of symmetry and regularity. Importantly, small neighborhood diversity can occur even in dense graphs, a property frequently observed in domains such as social networks with modular communities or knowledge graphs with repeated relational patterns. Thus, algorithms designed to work efficiently for graphs with small neighborhood diversity are capable of efficiently solving \mcs in complex settings where small vertex covers may not exist.

    We develop an algorithm with running time $vc^{\bigoh(vc)}\cdot\text{Poly}(n,c)$, and another algorithm with runtime $nd^{\bigoh(nd)}\cdot\text{Poly}(n,c)$. In the language of parameterized complexity, this implies that \mcs is fixed-parameter tractable (\fpt) parameterized by the vertex cover number and the neighborhood diversity. Notably, our algorithms remain efficient for arbitrarily many colors, as their complexity is polynomially dependent on the number of colors.
  \end{abstract}

  \pagebreak

  \input{1-introduction.tex}

  \input{2-notations.tex}

  \input{3-vertex_cover.tex}

  \input{4-neighborhood_diversity.tex}

  \input{5-acknowledgments.tex}

  \bibliographystyle{alpha}
  \bibliography{references.bib}
\end{document}

%% file: 1-introduction.tex
\section{Introduction}\label{sec:intro}
Clustering lies at the heart of numerous tasks in computer science and machine learning. In its essence, given a set of points in a metric space $(P,d)$, the objective is to partition them such that ``proximate'' points reside within the same cluster. While various unsupervised approaches exist, supervised learning offers a powerful paradigm for achieving ``most appropriate'' clustering.\smallskip

In supervised clustering, a labeled training dataset i.e., a subset of points $P' \subset P$ endowed with a coloring function $\colorof: P' \rightarrow [c]$ (where each color denotes a class/cluster) is provided to distill underlying patterns. Usually, given the training dataset, a learning algorithm outputs a set of cluster centers $C = \{c_1, \dots, c_r\}$. Subsequently, an unlabeled point $q$ is assigned the color $\colorof(c_i)$ where $c_i = \NN(q, C)$, with $\NN(q, C)$ representing the nearest neighbors of $q$ in $C$.\smallskip

However, the ever-increasing volume of modern datasets poses significant computational challenges for learning algorithms. Large datasets, while information-rich, often lead to protracted learning times. This has motivated a rich line of work on \emph{instance selection}, where the goal is to extract a small, yet representative, subset of the training data that preserves classification behavior. A classical formulation of this idea is the \emph{Minimum Consistent Subset} (\mcs) problem, introduced in 1968~\cite{hart1968}. Given a colored training dataset $T$, the \mcs problem seeks a minimum cardinality subset $S \subseteq T$ such that for every point $t \in T$, the color of $t$ is same as the color of at least one of its nearest neighbors in $S$. Despite its apparent simplicity, the \mcs problem poses significant computational hurdles and is known to be computationally hard in Euclidean spaces \cite{Wilfong, Khodamoradi}, and also hard to approximate in general settings~\cite{Chitnis22}.\smallskip

The \mcs problem has recently been extended to \emph{graph metrics}, motivated by applications where similarity is naturally modeled by graphs, such as social or knowledge networks. In the \emph{Consistent Subset Problem on Graphs} (\cspg), we are given a graph $G = (V, E)$ with a vertex coloring $\colorof: V \rightarrow [c]$. The distance metric is defined as the shortest path distance in $G$, denoted by $d(u, v)$. For a vertex $v \in V$ and a subset $U \subseteq V$, let $d(v, U) = \min_{u \in U} d(v, u)$. The set of nearest neighbors of $v$ in $U$ is denoted by $\NN(v, U) = \{u \in U : d(v, u) = d(v, U)\}$.\smallskip

A subset of vertices $S \subseteq V(G)$ is called a \emph{consistent subset} for $(G, \colorof)$ if, for every vertex $v \in V(G)$, the color of $v$ is present among the colors of its nearest neighbors in $S$, i.e.,  $\colorof(v) \in \colorof(\NN(v, S))$.

\begin{tcolorbox}[mydefstyle={\sc{Consistent Subset Problem on Graphs} (\cspg)}]
    \textbf{Input:} A graph $G$ and a coloring function $\colorof:V(G)\rightarrow [c]$.\\
    \textbf{Question:} Compute a minimum consistent subset for $(G,\colorof)$.
\end{tcolorbox}

This graph-theoretic version of the \mcs problem (i.e., \cspg) has recently drawn attention for both its theoretical appeal and practical relevance. Polynomial-time algorithms have been discovered for certain special graph classes, such as paths, spiders, and caterpillars~\cite{DeyMN23}, and later for bi-colored trees \cite{DeyMN21} and for $k$-colored trees (for fixed $k$) \cite{arimura2023minimum}. For more related works, we refer to \cite{abs-2405-14493, abs-2405-18569, BiniazK24}. Although results were known for bi-colored and $k$-colored trees, the status for the problem when the underlying graph is a general tree was open for a long time. In a recent breakthrough,~\cite{banik2024minimum} the authors systematically investigated \cspg and resolved this question. Their work led to two major contributions:

\begin{itemize}
    \item They established that \cspg is \npc on trees, resolving a key open question. This result is particularly striking, as many hard graph problems become tractable when restricted to trees.
    \item They designed a fixed-parameter tractable (\fpt) algorithm, i.e., an algorithm running in time $f(k) \cdot n^{\bigoh(1)} $, where $ k $ is a chosen parameter and $f$ is a computable function independent of the input size, for trees, with a running time of $\bigoh(2^{6c} \cdot n^6)$, where $c$ is the number of colors (chosen as the parameter) and $n$ is the number of vertices. This significantly improves upon earlier brute-force approaches with super-exponential dependence on $c$.
\end{itemize}

The hardness of \cspg on trees, where the minimum feedback vertex set (FVS) is empty set, has significant implications: it precludes the existence of an \fpt algorithm parameterized solely by FVS. This calls for stronger structural parameters to recover tractability. In this work, we take this challenge head-on and present FPT algorithms for \mcs parameterized by two natural and well-studied graph parameters:

\begin{itemize}
    \item \textbf{Vertex Cover Number ($vc$)}, which measures the minimum number of vertices needed to cover all edges of the graph.
    \item \textbf{Neighborhood Diversity ($nd$)}, which bounds the number of types of neighborhoods across the graph and is strictly stronger than vertex cover in dense graphs.
\end{itemize}

A formal definition for both parameters is presented in the next section. Our key contribution is that our algorithms are independent of the number of colors $c$, in stark contrast to prior work where the exponential dependence on $c$ was unavoidable. Specifically, we show:

\begin{itemize}
    \item \mcs admits an \fpt algorithm parameterized by vertex cover size, with running time $k^{\bigoh(k)} \cdot \text{poly}(n,c)$, where $f$ is color-independent and $k$ is the size of the vertex cover.
    \item \mcs also admits an \fpt algorithm parameterized by neighborhood diversity, again avoiding any exponential dependence on the number of colors.
\end{itemize}

In particular, we want to bring to the reader's attention that while designing an \fpt algorithm with dependence on both neighborhood diversity ($r$) and the number of colors ($c$) is straightforward, due to Claim~\ref{claim:none-one-all-vertices-from-a-type-color}, removing the dependence on the number of colors is highly non-trivial. This is because when the number of colors is large, the number of possible combinations becomes prohibitively high, resulting in a running time that is no longer \fpt in $r$. However, our key insight is that the interaction of a small number of important or \emph{responsible} colors with the solution is sufficient to determine the interaction of all other colors. While we may not be able to explicitly identify these responsible colors in advance, once we know how they interact with the solution, we can use a color-coding technique to probabilistically isolate and identify a most suitable set of such colors. This allows us to reduce the problem to a collection of independent subproblems, each of which can be solved separately using a greedy algorithm. The solutions to these subproblems can then be combined to obtain a solution to the original instance. To achieve this, we exploit structural properties arising from both the neighborhood diversity of the graph and the specific characteristics of the problem.\smallskip

At a high level, our algorithmic technique departs from the conventional use of color coding. Typically, color coding is employed to mark objects or structural features of a problem instance in a way that enables their independent resolution. In contrast, our approach involves color coding the elements themselves (in our case, the colors), with the goal of ensuring that a \emph{greedily selected} subset of solution elements remains well separated under the resulting color distribution. We believe that this perspective introduces a novel and potentially widely applicable direction for color coding, with possible extensions to a broader class of combinatorial problems beyond the specific context addressed in this work.

The parameter neighborhood diversity ($nd$) is particularly relevant in the context of AI and machine learning applications on graphs. While vertex cover captures a notion of "sparseness" around edges, neighborhood diversity provides a finer-grained measure of structural regularity. Graphs with small neighborhood diversity are those where most vertices have neighborhoods that are structurally similar, even if the graph is dense. Such structures appear in various real-world networks, including social networks with distinct community structures, or knowledge graphs where entities often share common relational patterns. An \fpt algorithm parameterized by $nd$ is significant because it indicates tractability not just for sparse graphs (like those with small vertex cover), but also for certain types of dense graphs that exhibit high regularity in their local connectivity patterns, a characteristic often observed in complex systems modeled as graphs in AI. This allows for efficient solutions in scenarios where a small vertex cover might not exist, but the underlying structure still permits algorithmic leverage.

%% file: 2-notations.tex
\section{Notations and Definitions}\label{sec:notations}

\noindent\textbf{Graph Notations and Definitions:} 
Let $G$ be a graph. We use $V(G)$ and $E(G)$ to denote the set of vertices and edges of $G$, respectively. For a set of vertices $S$, by $G \setminus S$ we mean $G[V(G) \setminus S]$, i.e., the subgraph of $G$ induced on $V \setminus S$. For a vertex $v$, $N(v)$ denotes the set of neighbors of $v$ in $G$   and $N[v] = N(v) \cup \{v\}$ denotes the closed neighborhood of $v$. We call a graph $G$ a {\em complete} graph if every pair of vertices in $G$ is adjacent. A {\em clique} in $G$ is an induced subgraph that is complete. In contrast, a set $I \subseteq V(G)$ is an {\em independent set} if no two vertices in $I$ are adjacent in $G$. A set $M \subseteq V(G)$ is a {\em vertex cover} if for every edge in $G$, at least one of its endpoints lies in $M$.\smallskip

Two vertices $u$ and $v$ are of the \emph{same type} if $N(v) \setminus \{u\} = N(u) \setminus \{v\}$. Note that, this defines an equivalence relation on $V(G)$~\cite{matsumoto2024space}. A \emph{neighborhood decomposition} of a graph $G$ is a partition $\mathcal{C} = \{C_1, C_2, \ldots, C_w\}$ of $V(G)$ such that all vertices in each $C_i$ are of same type. Each $C_i$ is a \emph{neighborhood class}, and $w$ is the size of the decomposition. The \emph{neighborhood diversity}, $\nd(G)$, is the minimum size of a neighborhood decomposition of $G$.

\begin{observation}
    Given a graph $G$, $\nd(G)$ can be computed in polynomial time~\cite{Lampis12}.
\end{observation}

We define the set of vertices at distance $\ell$ from a vertex $v$ by $\DistSet{v}{\ell} = \{ u \in V : d(u, v) = \ell \}$ and the set of vertices at distance $\ell$ from a vertex $v$ of color $a$ by $\DistSeta{v}{\ell}{a} = \{ u \in V : d(u, v) = \ell \text{ and } \colorof(u) = a \}$. For any vertex $v$, let $N_a(v)$ denotes the set of neighbors of $v$, with color $a$. For $X \subseteq V(G)$, we define $N(X)$ to be the neighbors of vertices in $X$. Most of the symbols and notations of graph theory used are standard and taken from \cite{Diestel12}.\smallskip

\noindent\textbf{Parameterized Complexity:} Parameterized complexity offers a framework for solving \np-hard problems more efficiently by isolating the combinatorial explosion to a parameter that is small in practice. A problem is fixed-parameter tractable (\fpt) if it can be solved in time $f(\ell) \cdot |I|^{\bigoh(1)}$, where $\ell$ is the parameter, $|I|$ is the input size, and $f$ is a computable function. Safe reduction rules are polynomial-time preprocessing steps that simplify the instance without changing its answer. For a detailed background, readers can refer to~\cite{CyganFKLMPPS15}.\smallskip

\noindent\textbf{Hitting Set:}
Given a set system $(\U,\F)$,  we say that $\Hset \subseteq \U$ is a {\em hitting set} for $(\U, \F)$ if $~\forall F \in \F,  ~\Hset \cap F \neq \emptyset$ and a set of subsets $\F' \subseteq \F$ is called a {\em set cover} for $(\U, \F)$ if $\bigcup_{F \in \F'} F = \U$. From \cite{CyganFKLMPPS15}[Theorem 6.1], we have the following proposition.

\begin{proposition}\label{prop:set-cover-runtime} Given a hitting set instance $\HS(\U, \F)$, a hitting set of minimum size can be found in time $2^{|\F|}(|\U| + |\F|)^{\bigoh(1)}$.
\end{proposition}

The $\bigohstar$ notation suppresses polynomial factors in the input size. Specifically, $\bigohstar(f(n)) = \bigoh(f(n) \cdot \text{poly}(n))$, where polynomial factors are omitted for clarity when they are not the focus of the analysis.

%% file: 3-vertex_cover.tex
\section{\fpt Algorithm Parameterized by Vertex Cover Size}

In this section, we present a fixed-parameter tractable (\fpt) algorithm for the \mcs problem parameterized by the size of the vertex cover. For completeness, we begin with a formal definition of the problem.

\begin{tcolorbox}[mydefstyle, title={\sc Consistent Subset Problem Parameterized by Vertex Cover Size}]
  \textbf{Input:} A graph $G = (V = M \sqcup I, E)$ where $|M| = k$ and $G[I]$ is an independent set, along with a coloring function $\colorof: V(G) \rightarrow [c]$.\\
  \textbf{Question:} Compute a minimum consistent subset (\mcs) $S$ for $(G, \colorof)$.\\
  \textbf{Parameter:} $k$
\end{tcolorbox}

It is well-known that the \textsc{Vertex Cover} problem is \fpt when parameterized by the solution size $k$~\cite{CyganFKLMPPS15}. Let $k$ be the size of the minimum vertex cover. As a preprocessing step, we compute a vertex cover $M$ of size $k$. We define $I = V(G) \setminus M$. Observe that the induced subgraph $G[I]$ is edgeless.

\begin{observation}\label{obs:diam-bounded-vc}
  For any two vertices $u$ and $v$ in $G$, $0 \leq d(u, v) \leq 2k$. In particular, if at least one of $u, v \in M$, then $0 \leq d(u, v) \leq (2k - 1)$.\label{obs:bound_path}
\end{observation}

\begin{proof}
 Let $P$ be a shortest path between vertices $u$ and $v$. Since $I$ is an independent set, no two consecutive vertices on $P$ can belong to $I$. Thus, between any two vertices from $I$, there must be at least one vertex from $M$.

 The path $P$ can contain at most $k$ vertices from $M$, as $|M| = k$. Therefore, the number of vertices from $I$ on $P$ is at most $k + 1$. This gives an upper bound on the total number of vertices in $P$ as $k + (k + 1) = 2k + 1$.

 In the case where either $u \in M$ or $v \in M$, the number of vertices from $I$ on $P$ can be at most $k$, and thus the total number of vertices in $P$ is at most $2k$. Therefore, the observation follows.
\end{proof}

Next, we make two guesses with respect to a minimum consistent subset $S$ and attempt to find a solution that respects the guesses.

\begin{description}
    \item[Guess 1:] We guess the distances from each vertex $u_{i}$ in $M$ to $S$. More specifically, we assume that an array $\D=[d_1,d_2,\cdots d_k]$ is given where $d_i$ denote the distance between $u_i$ and $S$. By Observation \ref{obs:bound_path}, each entry $d_i$ can take a value between $0$ and $(2k - 1)$. Thus, the total number of guesses for $\D$ is bounded by $(2k)^k$.
    \item[Guess 2:] 
    We guess the set of vertices $ M_1 \subseteq M $, which consists of the neighbors of the vertices $ S \cap I$. Formally, $M_1=\{u~|\text{ }u\in N(S \cap I)\setminus (S\cap M) \}$. The number of choices is bounded by $2^k$.
\end{description}

Let $I^{\textsc{out}}(\D)$ be the set of vertices in $I$  that are at a distance at most $d_i-1$ from some vertex $u_i\in M$. For any choice of $(\D,M_1)$, we say that a set of vertices $X\subseteq V(G)$ respects the choice $(\D, M_1)$, if $\forall u_{i} \in M$, $d(u_{i},X)=d_{i}$ and $N(X\cap I)=M_1$.
Therefore, given $(\D,M_1)$, our aim is to find a minimum cardinality consistent subset $S \subseteq V(G)$ that respects the choice $(\D,M_1)$.

\begin{observation}\label{obs:No-red-vertices}
  For any minimal consistent subset $S$ respecting $(\D,M_1)$, $S\cap I^{\textsc{out}}(\D)=\emptyset$.
\end{observation}

\begin{proof}
 Assume, for the sake of contradiction, that there exists a vertex $v \in S \cap I^{\textsc{out}}(\D)$. By the definition of $I^{\textsc{out}}(\D)$, there exists a vertex $u_i \in M$ such that $$d(u_i, v) \leq d_i - 1,$$ where $d_i = d(u_i, S)$ by definition. Since $v \in S$, it follows that $$d(u_i, S) \leq d(u_i, v) = d_i - 1,$$ which is a contradiction to $d_i = d(u_i, S)$ as defined in $\D$. Therefore, our assumption is false, and hence we conclude that $S \cap I^{\textsc{out}}(\D) = \emptyset.$
\end{proof}

We define $M_0=S\cap M$, i.e. $M_0=\{u_{i} \in M \text{ } | \text{ } d_{i}=0\}$ and $M_{x}= M \setminus (M_0 \cup M_1)$. Recall, for any vertex $v$, we denote the set of vertices at distance $d$ from $v$ by $\DistSet{v}{d}$, the set of vertices of color $a$ in the neighbor of $v$ by $N_a(v)$ and the set of vertices of color $a$ at distance $d$ from $v$ by $\DistSeta{v}{d}{a}$.\smallskip

We extend the scope of $\D$ and define $d_i$ for the vertices $u_{i}$ in $I$ as follows. Let $d_{u_{i}}^{min}$ be the minimum distance in $\D$ among the set of vertices $N(u_{i})$, i.e. $d_{u_{i}}^{min}=\min_{u_j\in N(u_{i})} d_j$. Note that all the neighbors of $u_i$ are in $M$ and hence $d_i=d_{u_{i}}^{min}+1$ is well defined. For any vertex $u_i \in I$, we define $C_i$ to be the set of colors of all those vertices that are at distance $d_{i}$ from $u_i$ and do not belong to the set $(M_{x} \cup M_{1})$, i.e. $C_{i}=\{\colorof(u_{j}) ~|~ u_j\in \DistSet{u_{i}}{d_i}\setminus (M_{x}\cup M_{1}\cup I^{\textsc{out}}(\D))\}$. Let $I^{\textsc{in}}\subseteq I$ be the set of vertices $u_i$ such that $\colorof(u_i)\notin C_{i}.$

\begin{observation}\label{obs:include-inconsistent-vertex}
    For any consistent subset $S$ respecting $(\D,M_1)$, $I^{\textsc{in}}(\D)\subseteq S$.
    \label{obs:mustinclude}
\end{observation}

\begin{proof}
   Suppose not. Let $u_i\in I^{\textsc{in}}(\D)$ but $u_i \notin S$.
   Also, let $x$ be the closest vertex in $S$ from $u_i$ such that $\colorof(x)=\colorof(u_i)$. 
   Consider $P$ as the shortest path between $u_i$ and $x$, also let $u_j \in M$ be the vertex next to $u_i$ in path $P$.
   Observe that, $d(u_j, x)< (d_{u_i}^{min}+1)-1=d_{u_i}^{min}\leq d_j$, which contradicts the assumption that $S$ respects the choice $\D$.
\end{proof}

\begin{observation}\label{obs:i-in-construction}
    Given $\D$, in polynomial time, we can find out the set of vertices in $I^{\textsc{in}}(\D)$ and $I^{\textsc{out}}(\D)$. 
\end{observation}

\begin{proof}
 For a given choice of $\D$, both $I^{\textsc{out}}(\D)$ and $I^{\textsc{in}}(\D)$  can be constructed in polynomial time using shortest path algorithms.   
\end{proof}

We have established that for any consistent subset $S$ respecting $\D$, $I^{\textsc{in}}(\D)\subseteq S$ and $I^{\textsc{out}}(\D)\cap S=\emptyset$. If $I^{\textsc{in}}(\D) \cap I^{\textsc{out}}(\D) \neq \emptyset$, then we simply discard the guess $\D$.\smallskip

We denote a vertex $u_i$ to be {\it satisfied} if $\exists$ a vertex $v \in \DistSet{u_i}{d_i}\cap (M_{0} \cup I^{\textsc{in}}(\D))$ such that $\colorof(u_{i})=\colorof(v)$. If a vertex is not satisfied, we call it {\it unsatisfied} and let ${\overline{U}}$ denote the set of all unsatisfied vertices. For any color $a$, let ${\overline{U}}_a \subseteq {\overline{U}}$ denote the subset of vertices in ${\overline{U}}$ that are colored $a$.

Let $S$ be any solution that respects $(\D,M_1)$. For each color $a$, define $S_a \subseteq S \setminus \left(I^{\textsc{in}}(\D) \cup M_0\right)$ to be the set of vertices in $S$ of color $a$, excluding those in $I^{\textsc{in}}(\D)$ and $M_0$. Let $S_a' \subseteq I\setminus I^{\textsc{out}}(\D) $ be any set of vertices of color $a$ such that for every vertex $u_i \in {\overline{U}}_a$, $d(u_i,S_a')\leq d_i$.
        
\begin{lemma}
 The set $S' = (S \setminus S_a) \cup S_a'$ is a consistent subset respecting $(\D,M_1)$, when $S$ is consistent with $(\D,M_1)$.\label{lemma:exchangeandind}
\end{lemma}

\begin{proof}
\noindent
 For the sake of contradiction, suppose that the set $S'$ is not consistent. Then, by the definition, there exists at least one vertex $u_i \in V(G)$ such that $\colorof(u_i) \notin \colorof(\NN(u_i, S'))$. Note that $u_i\in \overline{U} $. Now, if $u_i \notin \overline{U}_a$ i.e., $\colorof(u_i)=b$ (say). In this case, as $S$ is a consistent subset and by the construction of $S', \text{ } \exists $ a vertex $u_j \in S'$ of color $b$ such that $d(u_i, u_j)=d_i$. Hence, $u_i \in \overline {U}_a$ and by the definition of $S_a'$, there exists a vertex $u_j$ of color $\colorof(u_i)$ such that $d(u_i, u_j) \leq d_i$. Hence, we have $d_i' = d(u_i, S') < d_i$. 

 Observe that $u_i \in I$; otherwise, the fact that $d(u_i, S') < d_i$ would imply that $S'$ does not respect \emph{Guess 1}. Let $u_b$ be any vertex in $\NN(u_i, S')$. Let $u_c \in M$ be the neighbor of $u_i$ in $M$ lying on the path from $u_i$ to $u_b$. Then, $d(u_c, u_b) < d_i - 1$.

 We know that $d_i=d_{u_i}^{min}+1$ and $d_c\geq (d_{u_i}^{min}+1)-1=d_i-1> d(u_c, u_b)$. Hence, $S_a'$ contains a vertex at distance at most $d_c-1$. Thus $S_a'$ contains a vertex from $I^{\textsc{out}}(\D)$, which contradicts the definition of $S_a'$, completing the proof.
\end{proof}

Therefore, from Lemma~\ref{lemma:exchangeandind}, given $\D$ and $M_1$, for each color $a \in \colorof(\overline{U})$, our objective reduces to independently computing a minimum-size set $S_a^* \subseteq I \setminus I^{\textsc{out}}(\D)$ of color $a$, such that for every vertex $u_i \in \overline{U}_a$, it holds that $d(u_i, S_a^*) \leq d_i$.\smallskip

Recall, we define the set of vertices at distance $\ell$ from a vertex $v$ by $\DistSet{v}{\ell} = \{ u \in V : d(u, v) = \ell \}$.

For any vertex $u_i\in {\overline{U}_a}$, let $M_1^a(u_i)\subseteq N^{d_i-1}{(u_i)} \cap M_1$ to be the set of vertices such that each vertex in $M^a_1(u_i)$ has at least one neighbor of color $a$ in $I\setminus I^{\textsc{out}}(\D)$. Formally,
\begin{equation*}
  M_1^a(u_i)=\{u_j\in N^{d_i-1}{(u_i)}\cap M_1:N_a(u_j)\cap(I\setminus I^{\textsc{out}}(\D)) \neq \emptyset\}   
\end{equation*}

The intuition behind the definition of $M_1^a(u_i)$ is as follows. In order to satisfy any unsatisfied vertex $u_i$ of color $a$, any solution must include at least one vertex $u \in I\setminus I^{\textsc{out}}(\D)$ of color $a$ where $u$ is a neighbor of a vertex in $M_1^a(u_i)$. We define the following set system with ground set $M_1$, $\mathcal{M}_a(\D,M_1)=\{M_1^a(u_i)\}$.\smallskip

Let $X_a^*\subseteq N(S_a)\cap M_1$ be the \emph{minimal} set of vertices such that every vertex in $S_a$ has a neighbor in $X_a^*$.
\begin{observation}
    $X_a^*$ must be a minimal hitting set for $\mathcal{M}_a(\D,M_1)$.
\end{observation}

Towards finding $S_a^*$, we make the following final guess:

\begin{description}
    \item[\textbf{Guess 3:}] For each color $a$, guess the minimal hitting set $X_a \subseteq N(S_a) \cap M_1$. The total number of such choices is bounded by $2^k$.
\end{description}

Given $X_a$, consider the set system $(I,\F(X_a))$ where for each vertex $u_i\in X_a$ we include set of vertices $N(u_i) \cap (I\setminus(I^{\textsc{in}}(\D) \cup I^{\textsc{out}}(\D)))$ of color $a$ as a subset in the family $\F(X_a)$ i.e. 
$$\F(X_a)=\{N(u_i) \cap (I\setminus (I^{\textsc{in}} \cup I^{\textsc{out}})) \cap \colorof^{-1}(a):u_i \in X_a\}$$

For any choice $X_a$, let $S(X_a)$ denotes the minimum hitting set for $(I,\F(X_a))$.

\begin{observation}
    $(S\setminus S_a)\cup S(X_a)$ is a solution.
\end{observation}
\begin{proof}
    By construction for any vertex $u_i\in \overline{U}_a$, $d(u_i,S_a)\leq d_i$. Observe that from Lemma~\ref{lemma:exchangeandind}, we know that $(S\setminus S_a) \cup S(X_a)$ is a solution. Note that $S_a$ is a hitting set for $(I,\F(X_a))$. This completes the proof.
\end{proof}
\begin{observation}
    Given $\D$ and $M_1$, we can find out  $S_a^* \subseteq I \setminus I^{\textsc{out}}(\D)$ in time $2^{\bigoh(k)}$. 
\end{observation}
\begin{proof}
    Observe that there are at most $2^k$ possible choices for $X_a$. For each choice of $X_a$, $\F(X_a)$ contains at most $|X_a|\leq k$ sets.  Since the \textsc{Hitting Set} problem is solvable in time $Poly(n)\cdot 2^{m}$ with $n$ variables and $m$ sets (By Proposition \ref{prop:set-cover-runtime}), $S_{a}^{*}$ can be found in $2^{\bigoh(k)}$ time.
\end{proof}

All the sets $\{S_a^* ~|~a \in \colorof(\overline{U})\}$ can be found in time at most $c\cdot poly(n)\cdot 2^{\bigoh(k)}$, leading to $\bigohstar(k^{\bigoh(k)})$ overall running time.

\begin{theorem}
 \mcs is \fpt parameterized by vertex cover number, admitting an algorithm running in time $\bigohstar(k^{\bigoh(k)})$, where $k$ is the size of the vertex cover.
\end{theorem}

%% file: 4-neighborhood_diversity.tex
\section{\mcs Parameterized by Neighborhood Diversity/Types of Vertices}

 We are given a graph $G = (V, E)$. Let $V = \bigsqcup_{i \in [r]} T_{i}$ be a neighborhood decomposition of a graph $G$ of minimum size. Note that, $\forall i \in [r]$, the induced graph $G[T_{i}]$ is either an independent set or a clique, and for distinct $i, j \in [r]$, either there is no edge between $T_{i}$ and $T_{j}$ (or) all possible edges exist between vertices in $T_{i}$ and $T_{j}$.\smallskip

\begin{tcolorbox}[mydefstyle, title={\sc Consistent Subset Problem Parameterized by Neighborhood Diversity}]
  \textbf{Input:} A graph $G = (V = \bigsqcup_{i \in [r]} T_{i}, E)$ where for each $u,v \in T_i$, $N(u) \setminus \{v\} = N(v) \setminus \{u\}$ along with a coloring function $\colorof: V(G) \rightarrow [c]$.\\
  \textbf{Question:} Compute a minimum consistent subset (\mcs) $S$ for $(G, \colorof)$.\\
  \textbf{Parameter:} $r$
\end{tcolorbox}\smallskip

We show that \mcs is \fpt parameterized by neighborhood diversity $r$. To that end, we prove the following claim, which we use in the correctness proof of our algorithm at the end of this section.

\begin{claim}\label{claim:none-one-all-vertices-from-a-type-color}
Given a graph $G = (V, E)$ with neighborhood diversity $r$ (i.e., $V = \bigsqcup_{i \in [r]} T_{i}$), there is an \mcs $S$ for $(G, \colorof)$ such that for each type $T_{i}$ and for each color $j$, the set $S$ has $0,1$ or all the vertices of color $j$ from $T_{i}$. Formally, $\forall i \in [r]$ and $\forall j \in [c]$, we have $|T_{i} \cap \colorof^{-1}(j) \cap S| \in \{0, 1, |(T_{i} \cap \colorof^{-1}(j))|\}$.
\end{claim}
\begin{proof}
  Let $S$ be an \mcs such that includes $\ell$ vertices from $T_i \cap \colorof^{-1}(j)$ where $\ell \notin \{0,1,|T_{i} \cap \colorof^{-1}(j)|\}$. Let $x,y$ be any two distinct \ vertices in $T_{i} \cap \colorof^{-1}(j) \cap S$, and let $z$ be a vertex in $(T_i \cap \colorof^{-1}(j)) \setminus S$. We claim that the set $S' := S \setminus \{x\}$ is also a consistent subset, thereby contradicting the minimality of $S$ as an \mcs{} for $(G, \colorof)$. To establish this, we demonstrate that every vertex $v$ that was \emph{consistent} with respect to $S$ remains consistent with respect to $S'$, by analyzing the following two cases.\medskip

  \noindent\textbf{Case 1: ($v \neq x$)}  

  Since $x$ and $y$ are of the same type and have the same color, we have,

  \begin{equation*}
    d(v, S \cap \colorof^{-1}(j)) = d(v, (S \setminus \{x\}) \cap \colorof^{-1}(j)) \quad \forall j \in [c]
  \end{equation*} 

  and, since $v$ is consistent w.r.t. $S$, $\forall j \in [c]$,  \begin{align*}
    d(v, S \cap \colorof^{-1}(\colorof(v))) & \leq d(v, S \cap \colorof^{-1}(j)) \\
    \Rightarrow d(v, (S \setminus \{x\}) \cap \colorof^{-1}(\colorof(v))) & \leq d(v, (S \setminus \{x\}) \cap \colorof^{-1}(j))
  \end{align*}

  making it also consistent w.r.t. $S \setminus \{x\}$.\smallskip

  \noindent\textbf{Case 2: ($v=x$)} 
  
  Observe that, \begin{align*}
    d(x, (S \setminus \{x\}) \cap \colorof^{-1}(j)) 
    &= d(z, (S \setminus \{x\}) \cap \colorof^{-1}(j)) \\
    &= d(z, S \cap \colorof^{-1}(j)) \quad \forall j \in [c].
  \end{align*}
  
  And $z$ being consistent w.r.t. $S$, immediately implies $x$ to be consistent w.r.t. $S\setminus\{x\}$. Therefore, $S \setminus \{x\}$ forms a strictly smaller consistent subset, contradicting the minimality of $S$.
\end{proof}

The above claim essentially states that there exists a minimum consistent subset (\mcs) that, for each color from any type, includes either 0, 1, or all vertices of that color. With this claim in place, we are now ready to present the first step of our algorithm.
\subsection{Step 1: Identifying the Nature of Responsible Colors}

We start by defining \emph{partitions} and sets of \emph{responsible colors} with respect to a potential \mcs $S$ below. Notice that while we may guess (i.e., generate all possible) partitions required for a desired \mcs, generating all sets of responsible colors may not be possible in \fpt time. We use a clever approach to bypass the exhaustive generation of responsible color sets, as described at the end of these definitions.\smallskip

\noindent\textbf{Partitions (w.r.t. an \mcs $S$):} We begin by guessing a partition $\mathcal{T}$ of the $r$ types into 3 sets, namely $\T_0$, $\T_1$, and $\T_2$ with respect to a potential \mcs $S$ for $(G, \colorof)$ as follows:

\begin{itemize}
    \item $\T_{0} = \{ T_i ~|~ i \in [r] \text{ and } T_{i} \cap S = \emptyset \}$
    \item $\T_{1} = \{ T_i ~|~ i \in [r] \text{ and } |\colorof(T_{i} \cap S)| = 1 \}$
    \item $\T_{2} = \{ T_i ~|~ i \in [r] \text{ and } |\colorof(T_{i} \cap S)| > 1 \}$
\end{itemize}

In other words, $\T_0$ is the set of types that contain no vertex from $S$, $\T_1$ is the set of types from which all vertices selected into $S$ are of the same color, and $\T_2$ is the set of types from which vertices of multiple colors are selected into $S$.\smallskip

\noindent\textbf{Responsible Colors (w.r.t an MCS $S$):}  
Given an \mcs $S$ and a corresponding 3-partition $\T$, a small inclusion-wise-minimal set of colors $\R$ is a set of \emph{responsible colors}  if and only if it satisfies the following.

\begin{itemize}
    \item For each type $T_i\in \T_1$,  $\R$ contains the color $\colorof(T_{i} \cap S)$.
    \item For each type in $\T_2$, the set $\R$ contains at least two distinct colors from $\colorof(T_{i} \cap S)$.
\end{itemize}

Observe that any set of responsible colors has size at most $2r$, due to the minimality property. Moreover, any such set is sufficient to determine the partition $\T$ of types. And, for a given $S$, there may exist multiple sets of responsible colors, possibly more than polynomially (in $n$) many and finding one may not even be possible in \fpt time. Nevertheless, let $R = \{\rc_1, \ldots, \rc_{k}\}$ denote an arbitrary set of responsible colors for $S$ where $k\leq 2r$. We prove the following property of a set of responsible colors which we use in the final correctness proof of our algorithms. The property is that basically for every vertex $v$, its closest distance to a solution vertex in $S$ can be determined (same as) by its closest distance to a solution vertex whose color is from the set of responsible colors.

\begin{claim}\label{resp}
For a set of responsible colors $R$ of $S$ and an arbitrary vertex $v$, $d(v,S\setminus (S\cap \colorof^{-1}(\colorof(v))))= d (v, S \cap (\cup_{j\in R\setminus \colorof(v)} \colorof^{-1}(j)))$. 
\end{claim}

\begin{proof}
 Let $z$ be a vertex in $S$ of a color other than $\colorof(v)$ such that the distance from $v$ to $z$ is minimized over all vertices in $S$ whose colors are different from that of $v$, i.e.,\\
 $d(v,z)=d(v,S\setminus (S\cap \colorof^{-1}(\colorof(v))))$. If $z \in \T_1$, then by definition of $\T_1$, we have $\colorof(z) \in R$, satisfying the claim. Otherwise, if $z \in T_i$ for some $T_{i}\in \T_2$ and $\colorof(z) \notin R$, then by the definition of responsible colors, there must exist a $y \in S\cap T_i$ of a different color (i.e., $\colorof(y) \in R$, $\colorof(y) \neq \colorof(v)$). But then, we have $d(v, z) = d(v, y)$, proving the claim .
\end{proof}

We reiterate that although we may not be able to decide on an $R$, we can guess whether the vertices corresponding to the colors in $R$ are included in the solution from each type as described below.\smallskip 

\noindent\textbf{Guessing solution occurrences (nature) of colors in $R$:}  
We guess the solution occurrence of each responsible color $\rc \in R$ in each type via a function $occ : R \rightarrow 2^{[r]}$, where $occ(\rc)$ is the set of types that have vertices in $S \cap \colorof^{-1}(\rc)$. A \emph{valid} occurrence function $occ$ must be consistent with the following partitioning requirements consistent with $S$ and $\T$.

\begin{itemize}
\item For each type $T_i \in \T_{0}$, there is no $j\in[k]$ such that $i\in occ(\rc_j)$.
    \item For each type $T_i \in \T_{1}$, there is precisely one $j\in[k]$ such that $i\in occ(\rc_j)$.
    \item For each type $T_i \in \T_{2}$, there are at least two colors $\rc_{j_1}, \rc_{j_2} \in R$ such that $i \in occ(\rc_{j_1})\cap occ(\rc_{j_2})$.
\end{itemize}

At the end of Step 1, we assume that we have correctly fixed a partition $\T$ (with respect to a potential \mcs $S$), along with a consistent and valid occurrence function $occ$ for some arbitrary set of responsible colors $R$ (for $S$).\smallskip

\begin{figure}[!h]
\centering
\includegraphics[width=0.7\columnwidth]{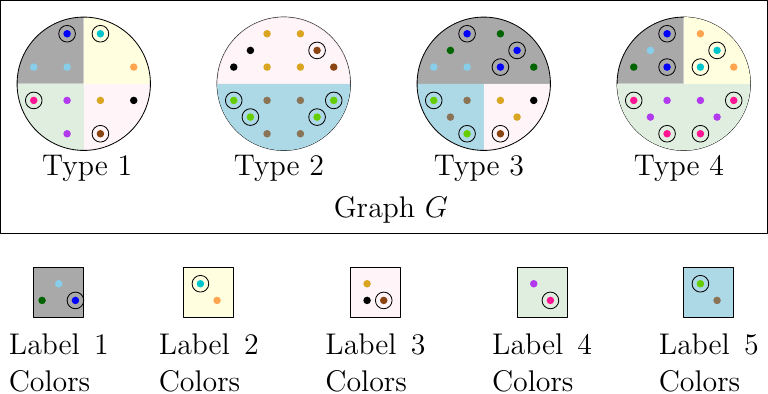} 
\caption{Each disk represents a type in the graph $G$. Colors are grouped into labels, and each level is indicated by a distinct background color. From each label, a representative color is selected, shown as a point encircled by a circle.}
\label{fig1}
\end{figure}

\subsection{Step 2: Label Coding to Identify a Most Suitable Set of Responsible Colors}
In this step, we apply a label-coding function that assigns each color in $R$ a distinct label with \emph{sufficiently high} probability. This allows us to break the problem into simpler subproblems, each of which is structurally easier, solvable in $f(r)\cdot n^{\bigoh(1)}$ time, and independent of the others. In each subproblem, we search for the most appropriate color that can assume the role of a responsible color from $R$. Since the input instance already associates colors with vertices, we use the term ``label coding'' instead of ``color coding'' to avoid confusion, although the two are essentially equivalent.\smallskip

Formally, this step aims to identify a set of actual colors from the input that can take on the roles of the guessed responsible colors, in a manner consistent with the guessed function $occ$, and compute the smallest possible consistent subset that realizes this correspondence. We proceed as follows.\smallskip

\noindent\textbf{Label Coding:} We label-code \cite{CyganFKLMPPS15} all the colors using $k$ labels and partition the color set $[c]$ as $[c] = C_{1} \uplus \cdots \uplus C_{k}$, such that with \emph{high} probability, each responsible color $\rc_i$ of $R$ gets the label $i$. We call such an event a \emph{nice} label-coding. Following a \emph{nice} label-coding, our goal becomes to identify the best responsible (a choice that gives the smallest possible consistent subset) color $c_j$s from each $C_i$ that aligns with the guessed/chosen partitioning constraints and $occ$ function.\smallskip

\noindent\textbf{Caution Constraints:} We select the most suitable responsible color for each $C_i$, where $i \in [k]$ in the next step. While selecting these most suitable responsible colors independently from each $C_i$,  we impose the following \emph{caution constraints} to ensure correctness. In our (desired) solution, in each $C_i$,

\begin{enumerate}
    \item[] C1: There is at least one color (denoted by $c_{i'}$) that has solution vertices precisely in all types of $occ(\rc_i)$.
    \item[] C2: There is exactly one color $c_{i'}$ that has solution vertices from any type $T_j\in \T_{1}$ where $T_j\in occ(\rc_i)$.
    \item[] C3: There is no color $c_{i'}$ that has solution vertices from any type in $ \T_{0}\cup (\T_1\setminus occ(\rc_i))$.
\end{enumerate}

These \emph{caution constraints} together ensure that, in the desired solution we aim to construct, the types corresponding to the selected vertices satisfy the identified (or guessed) partition $\T$.

\subsection{Step 3: Selection of a Best Responsible Color from Each Label with Caution Constraints}

To determine the most suitable responsible color from each $C_i$ and combine them to return an \mcs, we crucially exploit the fact that the subproblems of selecting responsible colors from each $C_i$ are \emph{independent}. At a high level, this independence arises because the partition $\T$, determined by the occurrence function $occ$ over $R$, essentially fixes, for every vertex, the distance of closest solution vertex of a different color. This, in turn, determines the \emph{minimum number of vertices} required from that particular color in the solution to ensure \emph{consistency} for all vertices of the same color. Since both the partition $\T$ and the function $occ$ are already fixed, we can compute, for each individual color, the smallest subset of vertices that must be included in a solution, as long as the $occ$ function requirements from $C_i$ and the \emph{caution constraints} are satisfied. Below we describe the exact procedure along with a formal correctness argument for the same.\smallskip

For a fixed label $C_{i}$, we go over each $c_{j} \in C_{i}$ expecting it to be a most suitable responsible color from $C_{i}$ and compute the size of a smallest set of vertices required to be in the solution for the consistency of all vertices of colors in $C_i$. First, from $occ(\rc_i)$, we determine the types from which vertices of color $c_j$ are to be selected into $S$. Recall from Claim \ref{claim:none-one-all-vertices-from-a-type-color}, either one or all vertices of color $c_j$ for each of the types in $occ(\rc_{i})$ are selected into a potential solution $S$. Let $O_{j}$ be the set of all possible subsets of vertices of color $c_{j}$ that may appear in $S$ in accordance with $occ(\rc_{i})$. Thus, $|O_{j}| \leq 2^{r}$. For a fixed $o \in O_{j}$, let $n'_{j, o}$ ($|o|$) be the number of vertices of color $c_{j}$ in $S$ and $n_{k, o}$ be the minimum number of vertices of color $c_{k}\in C_i\setminus\{c_{i}\}$ (again we have at most $3^k$ such choices) one has to pick into a solution of color $c_{k}$ adhering to caution constraints while satisfying the consistency requirement of all vertices of color $c_k$ and of all the vertices of color $c_j$ (with respect to  the choice $o$ ).\smallskip 

We formally check the consistency requirements as follows (in addition to caution constraints). For a choice $o$ of color $c_j$ and any arbitrary subset $o'$ of color $c_k$ (at most $3^k$ many such choices) that are to be selected into a potential solution, we must ensure that:
\begin{align}
d(v, o') &\leq \min\left\{ d(v, T_i) \mid T_i \in \mathcal{T}_1 \cup \mathcal{T}_2 \right\} \quad \forall v \text{ of color } c_k \label{eq:ck_consistency} \\
d(u, o) &\leq d(u, o') \quad \forall ~ u \text{ of color } c_j. \label{eq:cj_wrt_ck}
\end{align}

Equation~\eqref{eq:ck_consistency} ensures the consistency requirement for all vertices of color $c_k$, and Equation~\eqref{eq:cj_wrt_ck} ensures consistency of all vertices of color $c_j$ with respect to color $c_k$. Note that we do not have to worry about the \emph{consistency  requirements} between two colors $c_k$ and $c_k'$; since ensuring that each color is consistent with respect to a responsible set of colors (in \eqref{eq:ck_consistency})  is sufficient for it to be consistent with all the colors, due to Claim \ref{resp}.
 Let $S_i$ be a smallest subset of solution vertices  of colors in $C_{i}$ that satisfy the caution constraints along with the above mentioned consistency requirements, i.e.,

$$|S_i| = \min_{c_{j} \in C_{i}} \{ \min_{o \in O_{j}} \{ n'_{j, o} + \sum_{c_{k} \in C_{i} \setminus \{c_{j}\}} n_{k, o} \}\}$$

We return $S = \bigcup_{i \in k} S_{i}$ as the desired \mcs. Before presenting our final algorithm, we provide a correctness proof of the above statement by establishing the independence of the subproblems, specifically, that the selection of the best responsible color from each $C_i$ can be done independently. The following lemma essentially states that the smallest possible set of solution vertices of colors from $C_i$, satisfying the caution constraints and consistency requirements, can substitute the vertices of the same colors in an \mcs without violating consistency of any vertex or increasing the solution size.

\begin{lemma}\label{final}
For any \mcs $S$ with partition $\mathcal{T}$, $occ$, a set of responsible colors $R$, and following a nice label coding $[c] = C_{1} \uplus \cdots \uplus C_{k}$, let $S_i$ be a smallest possible set of vertices selected  from all  the colors in $C_i$ with $c_j$ being the responsible color, while adhering to the caution constraints and consistency requirements. $S'= S\cap (\cup_{j\notin C_i} \colorof^{-1}(j)) \cup S_i$ is also an \mcs.
\end{lemma}
\begin{proof}
  Let $v$ be an arbitrary vertex. We prove the consistency of $v$ w.r.t. $S'$ for the following two exhaustive cases.\smallskip 

  \noindent\textbf{Case 1:} ($\colorof(v)\notin C_i$) 
  Since $v$ is consistent w.r.t. $S'$, i.e., for any $j\in[c]$, we have,  \begin{align*}
  d(v, S \cap \colorof^{-1}(\colorof(v))) = d(v, S)
  \end{align*}

  But $d(v, S \cap  \colorof^{-1}(\colorof(v)))=d(v, S' \cap  \colorof^{-1}(\colorof(v)))$ as  $S$ and $S'$ contain the exact same set of vertices with color $\colorof(v)$. Moreover, for $S'$, we have a set of responsible colors $R'=(R\cup \{c_j\}) \setminus C_i $. Note that, $R$ and $R'$ differ by exactly one color and vertices of $c_j$ appear in $S_i$ exactly in the types as mapped by $occ$ function for the responsible color of $R$ from $C_i$ due to the caution constraints. And from Claim 1.11, 
  \begin{align*}
  d(v,S\setminus (S\cap \colorof^{-1}(\colorof(v)))) &= d (v, S \cap (\cup_{j\in R\setminus \colorof(v)} \colorof^{-1}(j)) )\\ 
  &=d(v,S'\setminus (S'\cap \colorof^{-1}(\colorof(v))))
  \end{align*}

  This together with the fact that $S$ and $S'$ contain the exact same set of vertices with color $\colorof(v)$ implies that  \begin{align*}
  d(v, S' \cap \colorof^{-1} (\colorof(v))) = d(v, S')
  \end{align*}

  \noindent\textbf{Case 2:} ($\colorof(v)\in C_i$) 
  Note that we have constructed $S_i$ adhering to the consistency requirements and caution constraints. Now for contradiction say there is a vertex of a different color vertex $u'$ such that $d(v,u')<d(v, S'\cap \colorof^{-1}(\colorof(v)))$. But then, there is also a vertex $u$ with a color from $R'$ such that $d(v,u)<d(v, S'\cap \colorof^{-1}(\colorof(v)))$. But this impossible due to the fact that consistency requirements (Eqn. (2)) if $\colorof(u)=c_j$ and (Eqn. (1)) if $\colorof(u)\neq c_j$) are maintained during our solution construction. 

  Hence, from both the cases, we have $S'$ is also an MCS.
\end{proof}

 Lemma \ref{final} ensures that one can compute each $S_i$ of minimum possible size from the corresponding label $C_i$, independently of the others, and combine them to obtain a desired Minimum Consistent Subset (\mcs). A formal algorithm is presented in Algorithm \ref{alg:mcs-nd}.\smallskip

\begin{algorithm}[tbh!]
\caption{\mcs parameterized by Neighborhood Diversity} \label{alg:mcs-nd}
\begin{algorithmic}[1]
  \STATE Generate all 3-partitions of the types into $\T_{0}, \T_{1}$, and $\T_{2}$.
  \STATE Generate all valid occurrence functions $occ$.
  \FOR{each fixed partition and valid $occ$}
    \STATE Label-code the colors $[c]$ using $k$ labels ($k \leq 2r$), and partition $[c] = C_{1} \uplus \cdots \uplus C_{k}$ based on the labels the colors receive.
    \FOR{$i = 1$ \TO $k$}
      \FOR{each $c_{j} \in C_{i}$}
        \STATE Let $c_{j}$ be the responsible color in $C_{i}$.
        \FOR{each $o \in O_{j}$}
          \STATE Compute $n'_{j,o}$ and the corresponding set of vertices (call it $S'_{j,o}$).
          \FOR{each $c_{k} \in C_{i}\setminus\{c_{j}\}$}
            \STATE Compute $n_{k,o}$ (as described in Step 3) and its vertex set $S'_{k,o}$.
          \ENDFOR
        \ENDFOR
      \STATE Keep track of the $S'_{j} = S'_{j, o} \cup (\bigcup_{c_{k} \in C_{i} \setminus \{c_{j}\}} S'_{k, o})$ which minimizes $n'_{j, o} + \sum_{c_{k} \in C_{i} \setminus \{c_{j}\}} n_{k, o}$.
      \ENDFOR
      \STATE $S_i \gets \arg\min_{S'_{j}: c_j \in C_i} |S'_j|$.
    \ENDFOR
    \STATE Keep track of $S \gets \bigcup_{i\in [k]} S_i$ of minimum cardinality.
  \ENDFOR
  \STATE \textbf{return} $S$.
\end{algorithmic}
\end{algorithm}

\noindent\textbf{Runtime Analysis:} The algorithm branches over $3^r$ partitions (choices for $\T$) and $r^{\bigoh(r)}$ possible $\mathit{occ}$ functions, each verifiable in polynomial time. A random labeling yields a \emph{nice} label-coding with probability at least $k^{-k}$, and such codings can be enumerated in $k^{\bigoh(k)} \cdot n^{\bigoh(1)}$ time. For each component $C_i$ and a responsible color $c_j$, there are $c \leq n$ choices, and at most $2^r$ options for $|O_j|$. For each $o \in O_j$, the value $n'_{j,o}$ can be computed in polynomial time. For each of the non-responsible colors $c_k$, values $n_{k,o}$ can be computed in $3^r \cdot \text{poly}(n)$ time by enumerating all $S \cap \colorof^{-1}(c_k)$. Thus, the total runtime is $3^r \cdot r^{\bigoh(r)} \cdot k^{\bigoh(k)} \cdot n^{\bigoh(1)} \cdot kc \cdot 3^r \cdot n^{\bigoh(1)} \cdot c \cdot 3^r \cdot n^{\bigoh(1)} = r^{\bigoh(r)} \cdot n^{\bigoh(1)}$, where the final bound follows from $k \leq 2r$ and $c \leq n$.\smallskip

The randomization step (label-coding) can be \emph{de-randomized with $(n,k)$-universal sets} ~\cite{CyganFKLMPPS15}, while maintaining the same asymptotic running time.

\begin{theorem} \mcs is \fpt parameterized by neighborhood diversity, admitting an algorithm running in time $\bigohstar(r^{\bigoh(r)})$, where $r$ is the neighborhood diversity of the input graph.
\end{theorem}

%% file: 5-acknowledgments.tex
\section{Acknowledgments}
The first author acknowledges support from the Science and Engineering Research Board (SERB) via the project MTR/2022/000253. The second author would like to thank Akanksha Agrawal (IIT Madras) for formally introducing him to the field of parameterized algorithms, a foundation that helped shape this work.